\documentclass{article}
\usepackage[margin=1in]{geometry}
\usepackage[onehalfspacing]{setspace}
\usepackage{amsmath,mathtools}
\usepackage{pgfplots}
\usepackage{subcaption}
\usepackage{booktabs}
\usepackage[braket, qm]{qcircuit}
\usepackage{graphicx}
\usepackage{authblk}
\usepackage{xcolor}
\usepackage[linesnumbered,ruled,vlined]{algorithm2e}
\usepackage{qcircuit}
\usepackage{booktabs}
\usepackage{cite}

\SetCommentSty{mycommfont}

\SetKwInput{KwInput}{Input}                
\SetKwInput{KwOutput}{Output}              
\SetKwProg{Fn}{Function}{}{}

\newcommand{\meqref}[1]{\text{Eq}.~\eqref{#1}}
\newcommand{\mref}[1]{Sec.\,\,\!$\ref{#1} $}
\newcommand{\mfig}[1]{Fig.\,\,\!$\ref{#1} $}

\usepackage[english]{babel}

\usepackage{pdflscape}
\usepackage{fancyvrb}
\usepackage{calc}
\usepackage{rotating}
\usepackage{pgf,tikz}
\usepackage{mathrsfs}
\usetikzlibrary{arrows}
\usepackage{mathtools}

\DeclarePairedDelimiter\floor{\lfloor}{\rfloor}
\usepackage{listings}
\usepackage[]{qcircuit}
\usepackage{braket}
\usepackage{mathtools}
\usepackage{varwidth}
\usepackage{caption}
\usepackage{subcaption}
\usepackage{graphicx}
\usepackage[export]{adjustbox}

\makeatletter
\def\paragraph{\@startsection{paragraph}{4}%
	\z@\z@{-\fontdimen2\font}%
	{\normalfont\bfseries}}
\makeatother

\usepackage{graphicx}
\usepackage{pgffor}
\usepackage{tikz,tikz-cd}
\usetikzlibrary{backgrounds,fit,decorations.pathreplacing}  
\usepackage{booktabs}
\usepackage{enumerate}

\usepackage{amssymb, amsmath, amsthm}

\usepackage{xypic}
\usepackage{scalerel}
\usepackage{ulem}

\usepackage{graphicx}              
\usepackage{amsmath}               
\usepackage{amsfonts}              
\usepackage{amsthm}                
\usepackage{xkeyval}               
\usepackage{amssymb}
\usepackage{enumerate}             
\usepackage{color}
\usepackage{breqn}                 
\usepackage{bm}                    

\allowdisplaybreaks[1]
\usepackage{nicefrac}
\usepackage{booktabs}

\usepackage{kpfonts}
\usepackage{stackengine}
\usepackage{calc}
\newlength\shlength
\newcommand\xshlongvec[2][0]{\setlength\shlength{#1pt}%
	\stackengine{-5.6pt}{$#2$}{\smash{$\kern\shlength%
			\stackengine{7.55pt}{$\mathchar"017E$}%
			{\rule{\widthof{$#2$}}{.57pt}\kern.4pt}{O}{r}{F}{F}{L}\kern-\shlength$}}%
	{O}{c}{F}{T}{S}}

\usepackage{hyperref}
\hypersetup{
	colorlinks   = true,    
	citecolor    = red      
}

\newcommand{\RN}[1]{%
	\textup{\uppercase\expandafter{\romannumeral#1}}%
}

\allowdisplaybreaks

\newtheorem{thm}{Theorem}[subsection]

\newtheorem{lem}[thm]{Lemma}

\newtheorem{cor}[thm]{Corollary}

\newtheorem{example}[thm]{Example}

\def\<{\langle}
\def\>{\rangle}

\numberwithin{equation}{section}

\newcommand{\abs}[1]{\lvert#1\rvert}

\begin{document}

	\title{A quantum algorithm for counting zero-crossings}

\author[1]{Alok Shukla \thanks{alok.shukla@ahduni.edu.in}}

\affil[1]{School of Arts and Sciences, Ahmedabad University, India}

	\date{}
	
	\maketitle

\begin{abstract}
We present a zero-crossings counting problem that is a generalization of the Bernstein–Vazirani problem. The goal of this problem is to count the number of zero-crossings (or sign changes) in a special type of sequence $ \mathcal{S} $, whose definition depends upon a secret string. A quantum algorithm is presented to solve this problem. The proposed quantum algorithm requires only one oracle query to solve the problem, whereas a classical algorithm would need at least $ n $ oracle queries, where $ 2^n $ is the size of the sequence $ \mathcal{S} $. In addition to solving the zero-crossings counting problem, we also give a quantum circuit for performing the Walsh-Hadamard transforms in sequency ordering. The Walsh-Hadamard transform in sequency ordering is used in a wide range of scientific and engineering applications, including in digital signal and image processing. Therefore, the proposed quantum circuit for computing the Walsh-Hadamard transforms in sequency ordering may be helpful in quantum computing algorithms for applications for which the computation of the Walsh-Hadamard transform in sequency ordering is required. 
\end{abstract}

	\section{Introduction}\label{sec:intro}

	The Bernstein–Vazirani algorithm \cite{bernstein1993quantum} is an example of a quantum algorithm that is superior to any classical algorithm in solving a special problem, the so-called Bernstein–Vazirani problem. We recall that the objective of the Bernstein–Vazirani problem  is to find the secret string $ s $ with $ s \in \{0,1\}^{n} $, assuming that a black-box oracle implementing the function $ f $ is given, where  $ {\displaystyle f\colon \{0,1\}^{n}\rightarrow \{0,1\}} $ is defined as $  f(x)=s\cdot x $. Here $ s \cdot x $ denotes the bit-wise dot product of  $s $ and $ x $ modulo $ 2 $, i.e.,  $s \cdot x =  s_{0}x_{0} + s_{1}x_{1}+ \ldots + s_{n-1}x_{n-1} \pmod 2 $, with $ s = s_{n-1}\,s_{n-2}\,\ldots \, s_1\, s_0 $, $ x =  x_{n-1}\,x_{n-2} \,\ldots \, x_1\, x_0 $, with $ s_i,\, x_i \in \{0,1\} $ for $ i=0,\,1,\, \ldots ,\,n-1$. 
	The Bernstein–Vazirani algorithm needs to query the oracle only once to find the secret string, whereas classically, at least $ n $ oracle queries are needed to determine the secret key.

	In this article, we present a special zero-crossings counting problem. This problem is a generalization of the Bernstein–Vazirani problem,  where the goal is to count the number of zero-crossings (or sign changes) in a special sequence whose definition depends upon a secret string. This special zero-crossings counting problem is precisely described in \mref{Sec:zero-crossing-problem}.  A quantum algorithm for solving the zero-crossings counting problem will be presented in \mref{sec:qunatum_solution}. 
	
	A slight modification of the quantum circuit used in solving the above-mentioned zero-crossings counting problem will allow us to obtain the Walsh-Hadamard transforms in sequency ordering. We note that usually the Walsh-Hadamard transform in natural order appears in quantum algorithms (for example, Deutsch-Jozsa algorithm \cite{deutsch1992rapid}, Bernstein–Vazirani algorithm \cite{bernstein1993quantum}, Simon's algorithm \cite{simon1997power}, Grover's algorithm \cite{grover1996fast}, Shor's Algorithm \cite{shor1999polynomial}, etc.), often to get a uniform superposition of quantum states at the beginning of the quantum algorithm.
	 The Walsh-Hadamard transform in natural order and the associated Walsh basis functions in natural order have also been used in the solution of non-linear ordinary differential equations \cite{SHUKLA2022127708}. On the other hand, the Walsh basis functions and the Walsh-Hadamard transforms in sequency ordering have found applications in several domains in engineering, for example, digital image and signal processing~\cite{kuklinski1983fast, zarowski1985spectral},  cryptography~\cite{lu2016walsh}, solution of non-linear ordinary differential equations and partial differential equations~\cite{beer1981walsh,ahner1988walsh,gnoffo2014global, gnoffo2015unsteady, gnoffo2017solutions}. In image processing applications, it is often desired to use  Walsh basis functions in sequency ordering because of better energy compaction properties. For instance, the hybrid classical-quantum algorithm approach for image processing discussed in \cite{Shukla2022} uses the Walsh-Hadamard transforms in sequency ordering. We note that in \cite{Shukla2022}, the quantum Walsh-Hadamard transform is performed in natural order, followed by classical computations to obtain the transforms in sequency ordering. In this and other applications wherein the Walsh-Hadamard transform in sequency ordering is needed, it is desirable to have a quantum circuit to compute the Walsh-Hadamard transforms directly in sequency ordering. In \mref{sec:Sequency-ordered-WH-transforms}, a quantum circuit for obtaining the Walsh-Hadamard transforms in sequency ordering will be presented and its working will be discussed. Finally, the conclusions will be summarized in \mref{sec:conclusion}.

	\subsection{Notation} \label{sec:notation}Before proceeding further, we define some convenient notations used in the rest of the paper.  
	\begin{itemize}
		\item $ \oplus $ : $ x \oplus y $ will denote $ x + y \mod 2 $.
	    \item $ s \cdot x $ : For $ s = s_{n-1}\,s_{n-2}\,\ldots \, s_1\, s_0 $ and $ x =  x_{n-1}\,x_{n-2} \,\ldots \, x_1\, x_0 $ with $ s_i,\, x_i \in \{0,1\}$, $ s \cdot x $ will denote the bit-wise dot product of  $s $ and $ x $ modulo $ 2 $, i.e.,  $s \cdot x :=  s_{0}x_{0} + s_{1}x_{1}+ \ldots + s_{n-1}x_{n-1} \pmod 2 $.
	    \item  $ s(m) $ : For  $ s =  s_{n-1}\,s_{n-2}\,\ldots\, s_2\,s_1\,s_0 $ and $ 1 \leq m \leq n $, $ s(m) $ denotes the string formed by keeping only the $ m $ least significant bits of $ s $, i.e., $ s(m) = s_{m-1}\,s_{m-2}\,\ldots\, s_2\,s_1\,s_0 $. 
	    \item On a few occasions,  by abuse of notations, a non-negative integer $ s $, such that  $ s = \sum_{j=0}^{n-1} \, s_j 2^j $, will be used to represent the $ n $-bit string
	    $ s_{n-1}\,s_{n-2}\,\ldots\, s_2\,s_1\,s_0 $.   
	\end{itemize}

\section{Zero-crossings counting problem} \label{Sec:zero-crossing-problem}
Let the function $ {\displaystyle F \colon \{0,1\}^{n}\rightarrow \{1,-1\}} $ be defined as $ F(x) = (-1)^{f(x)}$,  where $ f(x)= s \cdot x $ for some fixed secret string $ s \in \{0,1\}^{n}  $. The number of zero-crossings (i.e., sign changes) for the sequence 
		\begin{equation}\label{eq:def_sequence}
			\mathcal{S} = 	\left(\, F(0), \, F(1),\, F(2),\, \cdots \cdots ,\, F(2^n-1) \,\right)
		\end{equation}
is defined as 	\begin{equation}\label{Eq:defn_zero_crossings}
	\frac{1}{2} \sum_{k=0}^{N-2} \, \abs{   F(k+1) - F(k)}
	=  \frac{1}{2} \sum_{k=0}^{N-2} \, \abs{   (-1)^{s \cdot (k+1)} - (-1)^{s \cdot k}},
\end{equation}
where $ N =2^n $. 
The numbers of zero-crossings for sequences associated with different secret strings (for $ n=3 $) are listed in Table \ref{tab:zero-crosings-table}.
\begin{table}[]
	\centering
	\begin{tabular}{@{}ccc@{}}
		\toprule
		Secret string, s & Sequence $ \mathcal{S} = ( \, F(k) \, )_{k=0}^{7} $ & Number of zero-crossings (sign changes) \\ \midrule
		$ 000 $	& $ \displaystyle  \left( \, 1, \, 1, \, 1, \, 1, \, 1, \, 1, \, 1, \, 1        \, \right)$  &     $ 0 $                                  \\
		$ 001 $	& $ \displaystyle  \left( \,  1, -1,  1, -1,  1, -1,  1, -1 \, \right) $  &     $ 7 $                                  \\
		$ 010 $	& $ \displaystyle  \left(\,  1,  1, -1, -1,  1,  1, -1, -1 \, \right) $  &      $ 3 $                                \\
		$ 011 $	& $ \displaystyle  \left( \, 1, -1, -1,  1,  1, -1, -1,  1 \, \right) $  &      $ 4 $                                 \\
		$ 100 $	& $ \displaystyle  \left( \, 1,  1,  1,  1, -1, -1, -1, -1 \, \right) $  &      $ 1 $                                 \\
		$ 101 $	& $ \displaystyle  \left( \, 1, -1,  1, -1, -1,  1, -1,  1 \, \right) $  &      $ 6 $                                \\
		$ 110 $	& $ \displaystyle  \left( \, 1,  1, -1, -1, -1, -1,  1,  1 \, \right) $  &      $ 2 $                                 \\
		$ 111 $	& $ \displaystyle  \left( \, 1, -1, -1,  1, -1,  1,  1, -1   \, \right) $  &     $ 5 $                                     \\ \bottomrule
	\end{tabular}
	\caption{The table shows the number of zero-crossings (sign changes) for the sequence $ ( \, F(k) \, )_{k=0}^{7} $, with $ F(k) = (-1)^{s\cdot x} $, for $ s =  000, \, 001,\, \ldots,\, 111 $. }
	\label{tab:zero-crosings-table}
\end{table}
The computation of the number of zero-crossings for the sequence associated with the secret string  $ s=101 $ is illustrated in Ex.\,\,\!$ \ref{ex:one}$.
\begin{example} \label{ex:one}
	Let $ n = 3 $ and $ s = 101 $ (or equivalently $ s = 5$). 
	Clearly, 
	\begin{align*}
		s\cdot 0 &= (101)\cdot (000) = (1 \times 0 ) \, \oplus \,  (0 \times 0 ) \, \oplus \,  (1 \times 0 ) \,  =  0, \\
		s\cdot 1 &= (101)\cdot (001) = (1 \times 0 ) \, \oplus \,  (0 \times 0 ) \, \oplus \,  (1 \times 1 ) \,   =  1, \\
		s\cdot 2 &= (101)\cdot (010) = (1 \times 0 ) \, \oplus \,  (0 \times 1 ) \, \oplus \,  (1 \times 0 ) \,   =  0, \\
		\cdots & \cdots \\
		s\cdot 7 &= (101)\cdot (111) = (1 \times 1) \, \oplus \,   (0 \times 1) \, \oplus \,   (1 \times 1)  \, = 0. 
	\end{align*}
	Then $$ \mathcal{S} = \left((-1)^{s\cdot 0}, \,(-1)^{s\cdot 1},\, (-1)^{s\cdot 2},\, \ldots \,,\,   (-1)^{s\cdot 7} \right) =  \left( \, 1, -1,  1, -1, -1,  1, -1,  1 \, \right) .$$ 
	Let $ N= 2^n = 8 $. The number of zero-crossings of the sequence $ \mathcal{S} $ is given by 
	\[
	\frac{1}{2} \sum_{k=0}^{N-2} \, \abs{   (-1)^{s \cdot (k+1)} - (-1)^{s \cdot k}} = 
	\frac{1}{2} \left( \abs{-1 -1 } + \abs{1 - (-1)} + \abs{-1 - 1)} +  \abs{-1 - (-1)} + \abs{1 - (-1)}  + \abs{-1 -1} + \abs{1 - (-1)} \right) = 6.
	\]
\end{example}

Assuming that a black-box oracle implementing the function $ f(x) = {s \cdot x }$ for some fixed secret string $ s \in \{0,1\}^{n}  $ is given, the objective in the \textit{zero-crossings counting problem} is to find the number of zero-crossings for the associated sequence $ \mathcal{S} $ as defined in \meqref{eq:def_sequence}, making as few calls to oracle as possible. 
	
Clearly, any classical algorithm would require at least $ n $ oracle queries to find the number of zero-crossings. 
We note that one can use the Bernstein–Vazirani algorithm to find $s$; however, it will not directly give the number of  zero-crossings  and further classical computations will be needed to determine the number of  zero-crossings for the associated sequence $ \mathcal{S} $. We give a fully quantum solution for this problem in \mref{sec:qunatum_solution}, which requires only one oracle call to solve this problem and directly determines the number of zero-crossings for the sequence $ \mathcal{S} $. This quantum algorithm may be considered as the post-processing of the output of Bernstein-Vazirani algorithm. Indeed, this post-processing can also be carried out on a classical computer equally efficiently once $s$ is known. 
However, the primary contribution of this work lies in demonstrating that the same quantum circuit utilized for the post-processing of the Bernstein-Vazirani algorithm can also be leveraged to compute the Walsh-Hadamard transform in sequency ordering. As noted earlier, Walsh-Hadamard transform in sequency ordering is used in many scientific and engineering applications, including in digital signal and image processing,  cryptography and solution of non-linear ordinary differential equations and partial differential equations, etc.

\section{A quantum solution for zero-crossings counting problem} \label{sec:qunatum_solution}
In this section, we will describe our quantum circuit and algorithm (Algorithm~$ \ref{alg_Zero_Counting} $) for the solution of the zero-crossings counting problem. We will use Lemma \ref{lemma} to prove the correctness of Algorithm~$ \ref{alg_Zero_Counting} $.

\subsection{A quantum algorithm} \label{sec:algorithm}
Let $ s $ be fixed secret string of length $ n $, i.e., $ s \in \{0,1\}^{n}$. The goal is to find the zero-crossings of the sequence $ \mathcal{S} $ (as defined in \meqref{eq:def_sequence}). A schematic quantum circuit for the solution of this problem for an arbitrary $ n $  is given in \mfig{fig:full}. A complete quantum circuit for $ n=6 $ is given in \mfig{fig:full-fig-six} for illustration. We note that on excluding the last part of the circuit (separately shown in  \mfig{fig:last-part}), the remaining quantum circuit in \mfig{fig:full} is exactly the same as used in the original Bernstein–Vazirani algorithm. 
We also note that the action of oracle $ U_f $ on the input state $ \ket{x} \otimes \ket{y}  $ is given by
\begin{equation}\label{eq:oracle}
	U_f \ket{x} \otimes \ket{y} = \ket{x} \otimes \ket{y \oplus f(x)}.
\end{equation}
As noted in \mref{sec:notation}, the symbol $ \oplus $ denotes addition modulo $ 2 $, i.e., $ x \oplus y =  x + y \pmod 2 $.
The input in \mfig{fig:full} is prepared to be in the quantum state $ \ket{\psi_0} = \ket{0}^{\otimes n} \otimes \ket{1}   $, which means each of the first $ n $ qubits is in the $ \ket{0} $ state and the last ancilla qubit is in the state $ \ket{1} $. Application of $ H^{\otimes n} \otimes H $ on $ \ket{\psi_0} $ results in the state $ \ket{\psi_1}  = \left(H^{\otimes n} \otimes H \right)\ket{\psi_0} =  \frac{1}{\sqrt{N}}\sum_{k=0}^{N-1} \,  \ket{k} \otimes \ket{-}. $ Next the action of $ U_f $ on $ \ket{\psi_1}$ results in the state
\begin{equation}
	\ket{\psi_2} = U_f \ket{\psi_1} = \frac{1}{\sqrt{N}}\sum_{k=0}^{N-1} \, (-1)^{s \cdot k} \ket{k} \otimes \ket{-}. 
\end{equation}
\begin{figure}
	\centering
	\includegraphics[scale=1.05]{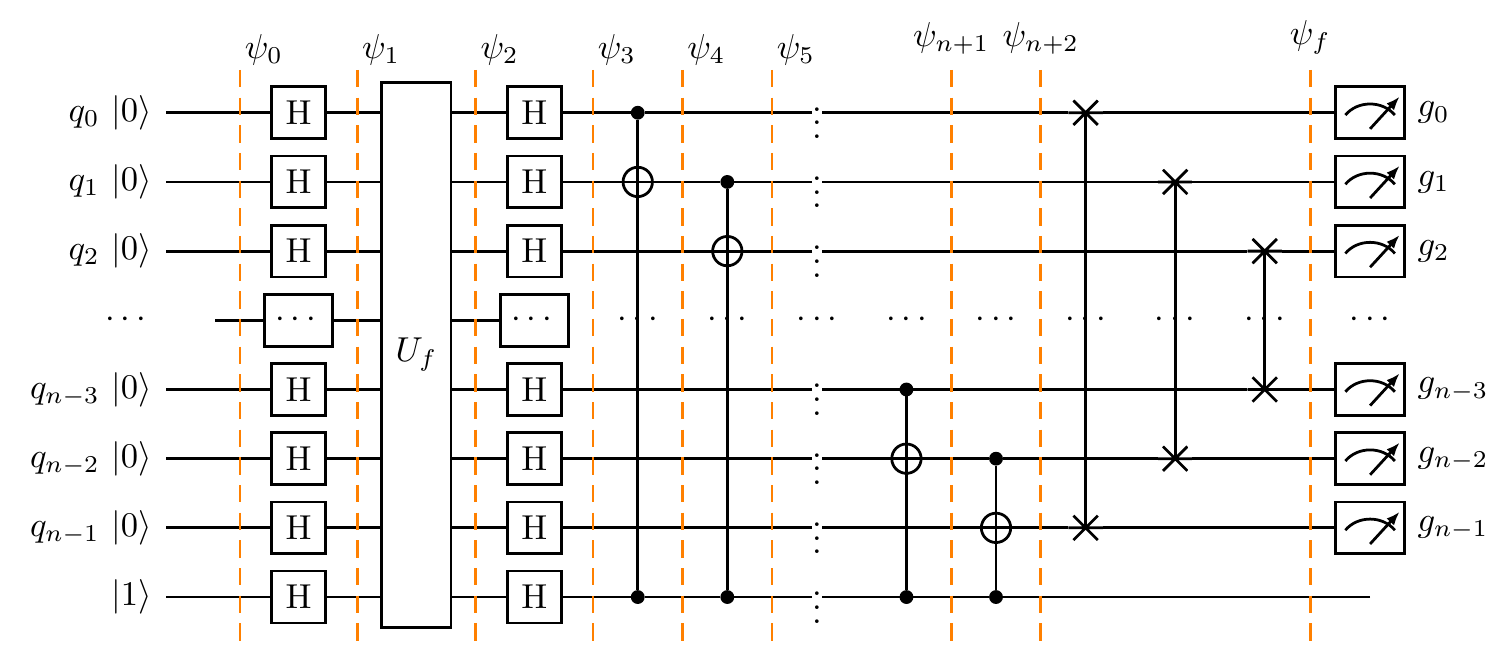}
	\caption{A schematic quantum circuit for counting zero-crossings for the sequence $ \mathcal{S} $ as defined in \meqref{eq:def_sequence}. All the $ \floor{\frac{n}{2}} $ swap gates at the end should be applied in parallel to reduce the circuit depth, but for clarity they are not shown to be in parallel in the figure.} \label{fig:full}
\end{figure}
Subsequently, an application of $ H^{\otimes n} \otimes H  $ on $ \ket{\psi_2} $ gives
\begin{align}\label{eq:Psi}
	\ket{\psi_3} = (H^{\otimes n} \otimes H ) \ket{\psi_2} &= \frac{1}{N} \, \sum_{k=0}^{N-1} \sum_{j=0}^{N-1} \, (-1)^{(  s \oplus j ) \cdot k} \ket{j} \otimes \ket{1} \nonumber\\
	&= \ket{s} \otimes \ket{1}.
\end{align}
The quantum circuit to obtain the state $ 	\ket{\psi_3} $ is the same as the quantum circuit in the original Bernstein–Vazirani algorithm and  \meqref{eq:Psi}  follows exactly as in the original Bernstein–Vazirani algorithm (clearly, if $ j =s $, then $ j \oplus s = 0 $, which means 
$ \frac{1}{ N }      \sum_{k =0}^{N-1} \, \left( (-1)^{(s \oplus j) \cdot k }  \right) \ket{j} \otimes \ket{1}  = \frac{1}{ N }      \sum_{k =0}^{N-1} \, \left( (-1)^{(0) \cdot k }  \right) \ket{s} \otimes \ket{1} =  \ket{s} \otimes \ket{1}$. Therefore, if  $ j \neq s $, then the corresponding coefficient of $ \ket{j}\otimes \ket{1} $ must be $ 0 $).

\begin{figure}
	\centering
	\includegraphics[scale=1.2]{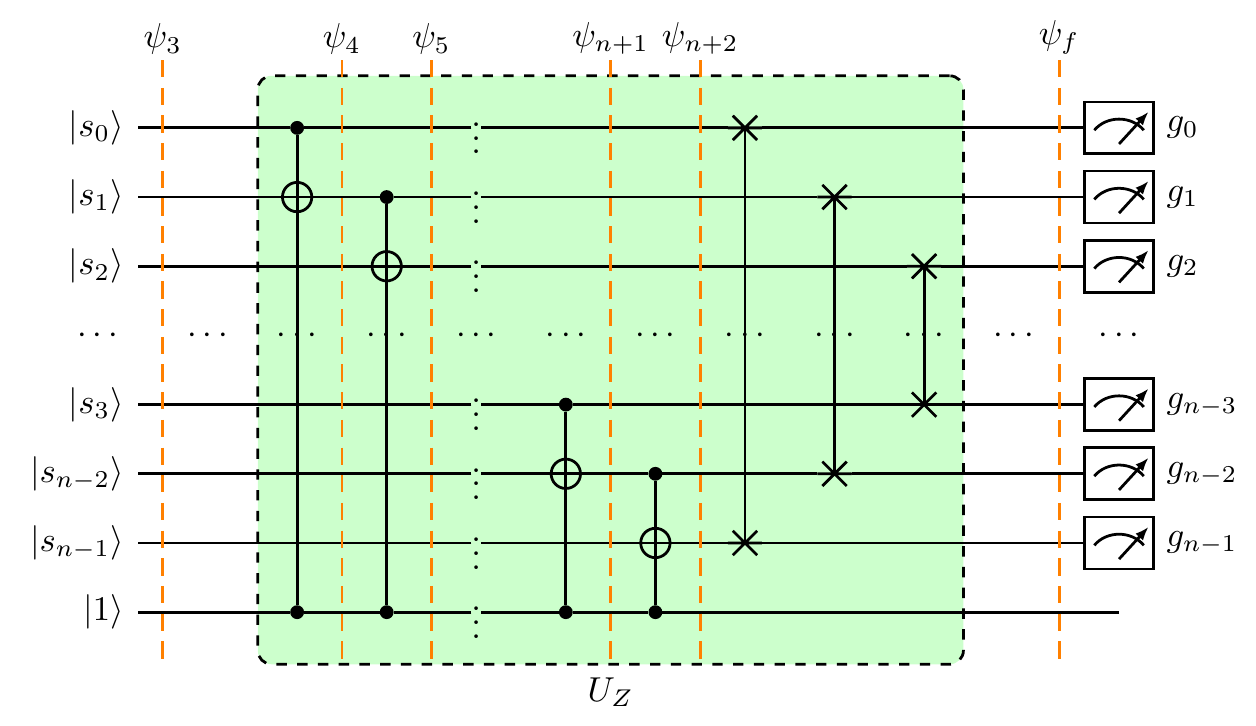}
	\caption{ Part of the quantum circuit for counting zero-crossings for the sequence $ \mathcal{S} $ as defined in \meqref{eq:def_sequence}. The quantum circuit $ U_Z $ consists of $ n-1 $ Toffoli gates and $ \floor{\frac{n}{2}} $ swap gates. } \label{fig:last-part}
\end{figure}

\begin{figure}
\centering
\includegraphics[scale=1.2]{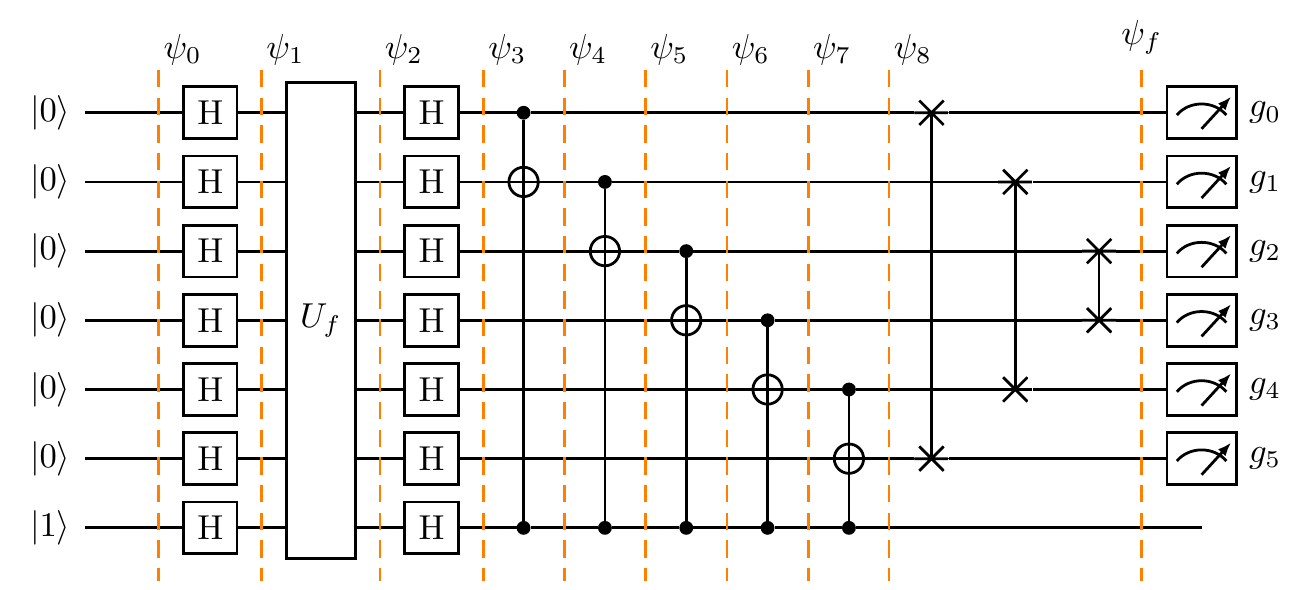}
\caption{Quantum circuit for counting zero-crossings for the sequence $ \mathcal{S} $ as defined in \meqref{eq:def_sequence} for $ n=6 $ qubits.} \label{fig:full-fig-six}
\end{figure}

Assume $ \ket{s} = \ket{s_{n-1}} \, \otimes \, \ket{s_{n-2}} \, \otimes \, \cdots \, \otimes \,  \ket{s_{1}} \, \otimes \ket{s_{0}}  $, i.e., assume the decimal representation of $ s $ to be  $ \sum_{m=0}^{n-1} \, s_m 2^m  $.
The next part of the quantum circuit in \mfig{fig:full} contains $ n-1 $ Toffoli gates (i.e., CCNOT gates) and $ \floor{\frac{n}{2}} $ swap gates. The swap gates should be applied in parallel. However, for clarity, they are not shown to be parallel in \mfig{fig:full}. The part of the quantum circuit consisting of Toffoli and swap gates is separately illustrated in \mfig{fig:last-part} and in the figure the action of Toffoli and swap gates is represented as $ U_Z $. The swap gates in \mfig{fig:last-part} simply reverse the order of qubits. The action of $ U_Z $ in the quantum circuit in \mfig{fig:last-part} can be described as 
 $ U_Z \ket{\psi_3} = \ket{g} \otimes \ket{1}$, where  \begin{equation}\label{eq:g}
 	g = \sum_{j=0}^{n-1} \, g_j 2^j,
 \end{equation}
and the bits of $ g $ are given by
\begin{align*}
	g_{n-1} & = \, s_0, \\
	g_{n-2} & =  \, s_0 \, \oplus \, s_1, \\ 
	\cdots&  \,\cdots \\
	g_1 &=  \, s_0 \, \oplus \, s_1 \, \oplus \, s_2 \, \oplus \, \ldots  \, \oplus \, s_{n-2}, \\
	g_0 &= \, \, s_0 \, \oplus \, s_1 \, \oplus \, s_2 \, \oplus \, \ldots  \, \oplus \, s_{n-2} \, \oplus \, s_{n-1}.
\end{align*}  
It means,
\begin{align}\label{eq:gk}
	g_{n-1} = s_0,  \quad \text{and} \quad	g_k &= \, \, s_0 \, \oplus \, s_1 \, \oplus \, \ldots   \, \oplus \, s_{n-1-k}, \quad \text{for $ k = n-2 $ to $ 0 $}.
\end{align}
 Finally, all the qubits (except the ancilla qubit) are measured and stored in a classical register. The classical bits resulting from these measurements are labeled $ g_0 $ to $ g_{n-1} $ in \mfig{fig:last-part}. Indeed,  $ g $ is number of zero-crossings of the sequence $ \mathcal{S} $ (as defined in \meqref{eq:def_sequence}). 
We recall that, the number of zero-crossings of the sequence $ \mathcal{S} $ (as defined in \meqref{eq:def_sequence}) is given by 
\[
 \frac{1}{2} \sum_{k=0}^{N-2} \, \abs{   (-1)^{s \cdot (k+1)} - (-1)^{s \cdot k}},
\]
as noted in  \meqref{Eq:defn_zero_crossings}.
It follows from Corollary \ref{cor} that $ g = Z_n(s) = \frac{1}{2} \sum_{k=0}^{N-2} \, \abs{   (-1)^{s \cdot (k+1)} - (-1)^{s \cdot k}}$ is the number of number of zero-crossings of the sequence $ \mathcal{S} $ that we wanted to determine.

The above discussion is captured in the following algorithm, i.e., Algorithm \ref{alg_Zero_Counting}. We note that Algorithm \ref{alg_Zero_Counting} always succeeds and requires only one query to the oracle $ U_f $ for computing the number of zero-crossings for the sequence $ \mathcal{S}$. \\

	\begin{algorithm}[H] \label{alg_Zero_Counting}
	\DontPrintSemicolon
	\KwInput{
		A black box oracle $ U_f $ that performs the transformation $ U_f \ket{x} \otimes \ket{y} = \ket{x} \otimes \ket{y \oplus f(x)} $ 	for $  x \in {0, \, \ldots \, ,\, N - 1} $, with $ N=2^n $ and  $ f(x)= s \cdot x $ for some fixed secret string $ s \in \{0, \, 1\}^{n}  $. }
	\KwOutput{The number of zero-crossings for the sequence $ \mathcal{S} = 	\left(\, F(0), \, F(1),\, F(2),\, \cdots \cdots ,\, F(2^n-1) \,\right) $, where $ F(x) = (-1)^{f(x)} $.
	}
	{
		$ \ket{\psi_0} = \ket{0}^{\otimes n} \otimes \ket{1}$. \tcp*{Initialization.}
		$  \ket{\psi_1} = \left(H^{\otimes n} \otimes H \right)\ket{\psi_0} =  \frac{1}{\sqrt{N}}\sum_{k=0}^{N-1} \,  \ket{k} \otimes \ket{-} $. \tcp*{Create superposition using Hadamard gates.} 
		$ \ket{\psi_2} = U_f \ket{\psi_1} = \frac{1}{\sqrt{N}}\sum_{k=0}^{N-1} \, (-1)^{s \cdot k} \ket{k} \otimes \ket{-} $. \tcp*{Apply $ U_f $.}
		$ \ket{\psi_3} = (H^{\otimes n} \otimes H ) \ket{\psi_2} = \frac{1}{N} \, \sum_{k=0}^{N-1} \sum_{j=0}^{N-1} \, (-1)^{(  s \oplus j ) \cdot k} \ket{j} \otimes \ket{1} = \ket{s} \otimes \ket{1}. $  \tcp*{Perform the Hadamard transform.} 
		$ \ket{\psi_f}  = U_Z \ket{\psi_3} = \ket{g} \otimes \ket{1}$.  \tcp*{Apply $ U_Z $ (ref.\,\,\!\mfig{fig:last-part}).}
		Measure the top $ n $ qubits to obtain $ g $.   \tcp*{$ g $ gives the number of zero-crossings for the sequence $ \mathcal{S}$.}
}
\caption{A quantum algorithm for computing the number of zero-crossings in the sequence	$ \displaystyle{\mathcal{S}} $ (ref.~\meqref{eq:def_sequence}).}
\end{algorithm}

The circuit depth (see \mfig{fig:full}) required for Algorithm \ref{alg_Zero_Counting} is $ \mathcal{O}(\log_2(N)) $, where $ N $ is the size of the sequence $ \displaystyle{\mathcal{S}} $. We also note that the swap gates in \mfig{fig:full} can be omitted from the circuit if one reads the bits of output $ g $ in the reverse order.

\subsection{A lemma} \label{sec:Lemma} 
 In this section, we will prove Lemma~\ref{lemma} and its corollary (Corollary \ref{cor}) to show that  $ g = \sum_{j=0}^{n-1} \, g_j 2^j $ (see  \meqref{eq:g}  and  \meqref{eq:gk}) gives the number of of zero-crossings in the sequence $ \mathcal{S} $ given by 
 \[
 \mathcal{S} = 	\left(\, F(0), \, F(1),\, F(2),\, \cdots \cdots ,\, F(2^n-1) \,\right),
 \]
 with $ F(x) = (-1)^{f(x)} $ and $ f(x) = s \cdot x $ (ref.\,\,\!\meqref{eq:def_sequence}). 
 
\begin{lem} \label{lemma}
	For an integer $ m $ with $ 1 < m \leq n $, and for 
	$ x = x_{n-1}\,x_{n-2} \,\ldots \, x_1\, x_0 $ (or equivalently, $ x $ with a decimal representation  $  x = \sum_{j=0}^{m-1} \, x_j 2^j  $),
	 with $ x_j \in \{0,1\} $,   define $ Z_m(x) $ as
	\begin{equation}\label{eq:def_Z_k}
		Z_m(x) := \frac{1}{2} \sum_{k=0}^{2^m-2} \, \abs{   (-1)^{x \cdot (k+1)} - (-1)^{x \cdot k}}.
	\end{equation}
Then, we have 
\begin{equation}\label{eq:lemma}
Z_m(s(m)) = 2 Z_{m-1} (s(m-1)) + (s_0 \, \oplus \, s_1 \, \oplus \, s_2 \, \oplus \, \ldots \, \oplus \,  s_{m-1}).	
\end{equation}
Here $ s(m) = s_{m-1}\,s_{m-2}\,\ldots \, s_1\, s_0 $
and $ s_0 \, \oplus \, s_1 \, \oplus \, s_2 \, \oplus \, \ldots \, \oplus \,  s_{m-1}= (s_0 + s_1 + s_2 + \ldots + s_{m-1}) \pmod 2 $. (ref.\,\,\!\mref{sec:notation}). 
\end{lem}

\begin{proof}
Let $ M = 2^m $. We put $ x=s(m) $ and split the summation on the right side of  \meqref{eq:def_Z_k}  into three different parts as follows.
\begin{align}	\label{eq:lem_all_terms}
	&Z_m(s(m)) \nonumber \\ &= \frac{1}{2} \left(\sum_{k=0}^{\frac{M}{2}-2} \, \abs{   (-1)^{s(m) \cdot (k+1)} - (-1)^{s(m) \cdot k}}\right) + \frac{1}{2} \abs{  (-1)^{s(m) \cdot \left(\frac{M}{2}\right)} - (-1)^{s(m) \cdot \left(\frac{M}{2} -1 \right) }} + \frac{1}{2} \sum_{k=\frac{M}{2}}^{M-2} \, \abs{   (-1)^{s(m) \cdot (k+1)} - (-1)^{s(m) \cdot k}}.
\end{align} 
We note that $ \frac{M}{2} = 2^{m-1}$ represents the $ m $-bit string $ 10\ldots0 $. Therefore, $ s(m) \cdot \left(\frac{M}{2}\right) = s_{m-1} $. Similarly,  $ \frac{M}{2} -1 = 2^{m-1}-1$ represents the $ m $-bit string $ 011\ldots1 $, hence  $ s(m) \cdot \left(\frac{M}{2} -1 \right) =  s_0 \, \oplus \, s_1 \, \oplus \, s_2 \, \oplus \, \ldots \, \oplus \,  s_{m-2}$.
Therefore, the middle term in the above summation reduces to 
\begin{align}\label{eq:lem_second_term}
	\frac{1}{2} \abs{  (-1)^{s(m) \cdot \left(\frac{M}{2}\right)} - (-1)^{s(m) \cdot \left(\frac{M}{2} -1 \right) }} & = \frac{1}{2} \abs{  (-1)^{s_{m-1}} - (-1)^{ s_0 \, \oplus \, s_1 \, \oplus \, \ldots \, \oplus \,  s_{m-2} }} \nonumber \\ & = s_0 \, \oplus \, s_1 \, \oplus \, s_2 \, \oplus \, \ldots \, \oplus \,  s_{m-1}.
\end{align}
Next we show that the last and the first terms are equal. 
\begin{align}\label{eq:lem_last_term}
\frac{1}{2} \sum_{k=\frac{M}{2}}^{M-2} \, \abs{   (-1)^{s(m) \cdot (k+1)} - (-1)^{s(m) \cdot k}} & = 	\frac{1}{2} \sum_{k=0}^{\frac{M}{2}-2} \, \abs{   (-1)^{s(m) \cdot (\frac{M}{2} + k+1)} - (-1)^{s(m) \cdot (\frac{M}{2} + k)}}  \nonumber \\
  & =   \frac{1}{2} \sum_{k=0}^{\frac{M}{2}-2} \, \abs{   (-1)^{s_{m-1}}  \left((-1)^{s(m) \cdot (  k+1)} - (-1)^{s(m) \cdot (k)}  \right)  } \nonumber \\
  & =  \frac{1}{2} \sum_{k=0}^{\frac{M}{2}-2} \, \abs{    \left((-1)^{s(m) \cdot (  k+1)} - (-1)^{s(m) \cdot (k)}  \right)  }. 
\end{align}
From \meqref{eq:lem_all_terms}, \meqref{eq:lem_second_term} and \meqref{eq:lem_last_term} it follows that
\begin{align}
Z_m(s(m)) & = \left( \sum_{k=0}^{\frac{M}{2}-2} \, \abs{   (-1)^{s(m) \cdot (k+1)} - (-1)^{s(m) \cdot k}}\right) + (s_0 \, \oplus \, s_1 \, \oplus \, \ldots \, \oplus \,  s_{m-1}) \nonumber \\
& =  \left( \sum_{k=0}^{\frac{M}{2}-2} \, \abs{   (-1)^{s(m-1) \cdot (k+1)} - (-1)^{s(m-1) \cdot k}}\right) + (s_0 \, \oplus \, s_1 \, \oplus \, \ldots \, \oplus \,  s_{m-1}).  
\end{align}
The last step follows because as $ k $ runs through $ 0$ to  $ \frac{M}{2} -2 = 2^{m-1} -2 $, the computations of  $ s(m) \cdot (k+1) $ and $ s(m) \cdot k $  involve only the $ m-1 $ least significant bits of $ s $, allowing one to write $ s(m) \cdot (k+1) = s(m-1) \cdot (k+1) $ and $ s(m) \cdot k = s(m-1) \cdot k $. 
Hence, we obtain
\begin{align*}
Z_m(s(m)) = 2 Z_{m-1} (s(m-1)) + (s_0 \, \oplus \, s_1 \, \oplus \, s_2 \, \oplus \, \ldots \, \oplus \,  s_{m-1}),
\end{align*}
and the proof is complete.
\end{proof}

\begin{cor}\label{cor}
	Let  $ s = s_{n-1}\,s_{n-2}\,\ldots \, s_1\, s_0 $ with $ s_j \in \{0,1\} $. If $ 	Z_n(s) $ is defined as 
	\[
		Z_n(s) := \frac{1}{2} \sum_{k=0}^{2^n-2} \, \abs{   (-1)^{s \cdot (k+1)} - (-1)^{s \cdot k}},
	\]
	then 
	\begin{equation}\label{eq:cor}
		Z_n(s) = \sum_{k=0}^{n-1} \, g_k 2^k,
	\end{equation}
where $ g_{n-1} = s_0 $ and  $ g_k =  s_0 \, \oplus \, s_1 \, \oplus \, \ldots \,  \oplus \, s_{n-1-k}\,$ for $ k=n-2$ to $ k=0 $.
\end{cor}
\begin{proof}
	It is easy to see from  \meqref{eq:lemma} that 
	\begin{align*}
		Z_1(s(1)) &  = s_0  = g_{n-1}\\ 
		Z_2(s(2)) & = 2 Z_1(s(1)) + (s_0 \, \oplus \,  s_1) = 2 g_{n-1} + g_{n-2} \\ 
		Z_3(s(3)) & = 2 Z_2(s(2)) + (s_0 \, \oplus \,  s_1 \, \oplus s_2) = 2^2 g_{n-1} + 2 g_{n-2} + g_{n-3}. 
	\end{align*}
A simple induction argument, whose details we skip, and the observation that $s(n) = s$  (see \mref{sec:notation}), shows that
\[
	Z_n(s) = \sum_{k=0}^{n-1} \, g_k 2^k.
\]
This completes the proof. 
\end{proof}
Clearly, $ Z_n(s) $ computes the number of zero-crossings of the sequence $ \mathcal{S} $ (ref.\,\,\!\meqref{eq:def_sequence} and  \meqref{Eq:defn_zero_crossings}).
In the following, we give an example to illustrate the steps for computing $ Z_n(s) $.
\begin{example}\label{ex:two}
	Let $ n = 3 $ and $ s = 5 $ (or equivalently $ s = 101 $ as a binary string). 
We have $$ \mathcal{S} = \left((-1)^{s\cdot 0}, \,(-1)^{s\cdot 1},\, (-1)^{s\cdot 2},\, \ldots \,  ,\, (-1)^{s\cdot 7} \right) =  \left( \, 1, -1,  1, -1, -1,  1, -1,  1 \, \right) .$$ 
Let $ N= 2^n = 8 $. 
As $ s = 101 $, we have $ s_0 = 1 $, $ s_1 =0 $ and $ s_2 = 1 $. The number of zero-crossings of the sequence $ \mathcal{S} $ is given by 
\begin{align}\label{eq:eaxample_one}
	Z_3(s)) = Z_3(s(3)) = 2 Z_2 (s(2)) + (s_0 \, \oplus \, s_1 \, \oplus \, s_2) = 2 Z_2(s(2)) + (1 \, \oplus \, 0 \, \oplus \, 1) = 2 Z_2(s(2)).
\end{align}
We have, $ s(2) = 01 $. Therefore,
\begin{equation}\label{eq:example_two}
	Z_2(s(2)) = 2 Z_1(s(1)) + (s_0 \, \oplus \, s_1 \,) = 2 Z_1(s_0) + ( 1 \, \oplus \, 0 \,) = 2 (1) + 1 = 3.
\end{equation}
From \meqref{eq:eaxample_one}  and  \meqref{eq:example_two}  we get $ Z_3(s)) = 6 $, which is the same result that we obtained in Ex.\,\,\!$ \ref{ex:one} $. 
\end{example}

\section{Sequency ordered Walsh-Hadamard transforms} \label{sec:Sequency-ordered-WH-transforms}
In this section, we will briefly recall the Walsh basis functions in sequency and natural ordering. We will also discuss the Walsh-Hadamard transforms in sequency and natural ordering and describe a quantum circuit for performing the Walsh-Hadamard transforms in sequency ordering. Interested readers may refer to \cite{beauchamp1975walsh} for further details on Walsh basis functions, Walsh-Hadamard transforms, and their applications. 

\subsection{Walsh basis functions in sequency and natural ordering}
Walsh basis functions $ W_k (x) $ for $  k =0,~1,~2, ~\ldots~ N-1 $  in sequency order are defined as follows
\begin{align}
	W_0(x) &= 1 \quad \text{for } 0 \leq x \leq 1,  \\
	W_{2k} (x) &= W_k(2x) + (-1)^k W_k (2x -1 ),  \\
	W_{2k+1} (x) &= W_k(2x) - (-1)^k W_k (2x -1 ), \\
	W_k(x) &= 0 \quad \text{for } x < 0 \text{ and } x >1,
\end{align}
where $ N $ is an integer of the form $ N = 2^n$. 
For $ N =8 $ the Walsh functions in sequency order are shown in \mfig{fig_walsh_sequency}. 

\begin{figure}[htbp]
	\begin{center}
			\includegraphics[trim=100 225 100 225, clip]{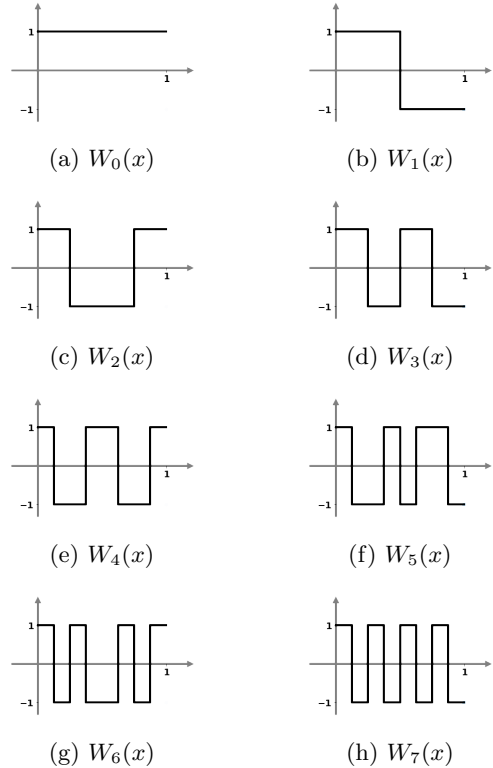}
		\caption{Walsh basis functions in the sequency ordering for $ N=8 $.} 	\label{fig_walsh_sequency}
	\end{center}
\end{figure}

The number of sign changes (or zero-crossings) for the Walsh functions increases as the orders of the functions increase. A vector of length $ N $ can be obtained by sampling a Walsh basis function. The Walsh-Hadamard transform matrix in sequency order is obtained by arranging the vectors obtained from sampling the Walsh basis functions as the rows of a matrix. The vectors are arranged in increasing order of sequency. The Walsh-Hadamard transform matrix of order $ N=8 $ in sequency order is
\begin{align*}
	\frac{1}{\sqrt{8}} \,
	\begin{pmatrix*}[r]
		1 & 1 & 1 & 1 & 1 & 1 & 1 & 1  \\
		1 & 1 & 1 & 1 & -1 & -1 & -1 & -1  \\
		1 & 1 & -1 & -1 & -1 & -1 & 1 & 1  \\
		1 & 1 & -1 & -1 & 1 & 1 & -1 & -1  \\
		1 & -1 & -1 & 1 & 1 & -1 & -1 & 1  \\
		1 & -1 & -1 & 1 & -1 & 1 & 1 & -1  \\
		1 & -1 & 1 & -1 & -1 & 1 & -1 & 1  \\
		1 & -1 & 1 & -1 & 1 & -1 & 1 & -1  \\
	\end{pmatrix*}.
\end{align*} 
In contrast, the Walsh-Hadamard transform matrix of order $ N =2^n $ in natural order is given by  
\begin{align*}
	H_N = H^{\otimes n},
\end{align*}  
where 
\begin{align*}
	H = \frac{1}{\sqrt{2}} \, \begin{pmatrix*}[r]
		1 & 1  \\
		1 & -1  \\
	\end{pmatrix*}.
\end{align*}
and $ N = 2^n $. 
\begin{figure}[htbp]
	\begin{center}
			\includegraphics[trim=100 225 100 225, clip]{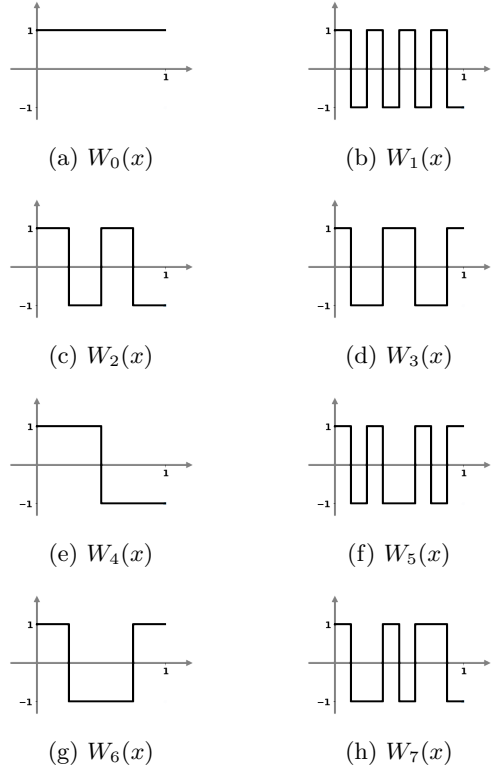}
		\caption{Walsh basis functions in the natural ordering for $ N=8 $.}	\label{fig_walsh_natural}
	\end{center}
\end{figure}
The Walsh-Hadamard matrix in natural order for $ N=8 $ is 
\begin{align*}
		\frac{1}{\sqrt{8}} \,
	\begin{pmatrix*}[r]
		1 & 1 & 1 & 1 & 1 & 1 & 1 & 1  \\
		1 & -1 & 1 & -1 & 1 & -1 & 1 & -1  \\
		1 & 1 & -1 & -1 & 1 & 1 & -1 & -1  \\
		1 & -1 & -1 & 1 & 1 & -1 & -1 & 1  \\
		1 & 1 & 1 & 1 & -1 & -1 & -1 & -1  \\
		1 & -1 & 1 & -1 & -1 & 1 & -1 & 1  \\
		1 & 1 & -1 & -1 & -1 & -1 & 1 & 1  \\
		1 & -1 & -1 & 1 & -1 & 1 & 1 & -1  \\
	\end{pmatrix*}.
\end{align*}
Walsh basis functions (or Hadamard-Walsh basis functions) in natural order can be obtained using the rows of the Walsh-Hadamard matrix $ H_N$. For $ N =8 $, the Walsh functions in natural order are shown in \mfig{fig_walsh_natural}.

The Walsh-Hadamard transforms in natural and sequency orders can also be defined in terms of their actions on the computational basis vectors. 
Let $ N=2^n $ be a positive integer. Let $ V$ be the $ N $ dimensional complex vector space generated by the computational basis states $ \{ \ket{0}, \, \ket{1}, \, \ldots \,,\, \ket{N-1} \} $.  We note that the Walsh-Hadamard transform in natural order can be defined as a linear transformation $ H_N : V \to V $ such that the action of $ H_N = H^{\otimes n}$ on the computational basis state $ \ket{j} $, with $ 0 \leq j \leq N-1 $ is given by
\begin{equation}\label{eq:natual:hadamard}
	H_N \, \ket{j} = \frac{1}{\sqrt{N}} \sum_{k=0}^{N-1} \, (-1)^{k \cdot j} \, \ket{k}. 
\end{equation}
Here $ k \cdot j $ denotes bit-wise dot product of $ k $ and $ j$. 

Next, we note that the Walsh-Hadamard transform in sequency order can be defined as a linear transformation $ H_S : V \to V $ acting on the  basis state $ \ket{j} $, with $ 0 \leq j \leq N-1 $, as follows (see \cite{beauchamp1975walsh}). 
\begin{equation}\label{eq:sequency:hadamard}
	H_S \, \ket{j} = \frac{1}{\sqrt{N}} \sum_{k=0}^{N-1} \, (-1)^{ \sum_{r=0}^{n-1} \, k_{n-1-r} (j_r \oplus j_{r+1}) } \, \ket{k},
\end{equation}
where $ k = k_{n-1}\,k_{n-2}\,\ldots \, k_1\, k_0 $ and $ j =  j_{n-1}\,j_{n-2} \,\ldots \, j_1\, j_0 $, are binary representations of $ k $ and $ j $ respectively, with $ k_i,\, j_i \in \{0,1\} $ for $ i=0,\,1,\, \ldots ,\,n-1$, and $ j_{n} = 0 $.

\subsection{Quantum circuit for performing the sequency ordered Walsh-Hadamard transforms}%

\begin{figure}
	\centering
	\includegraphics{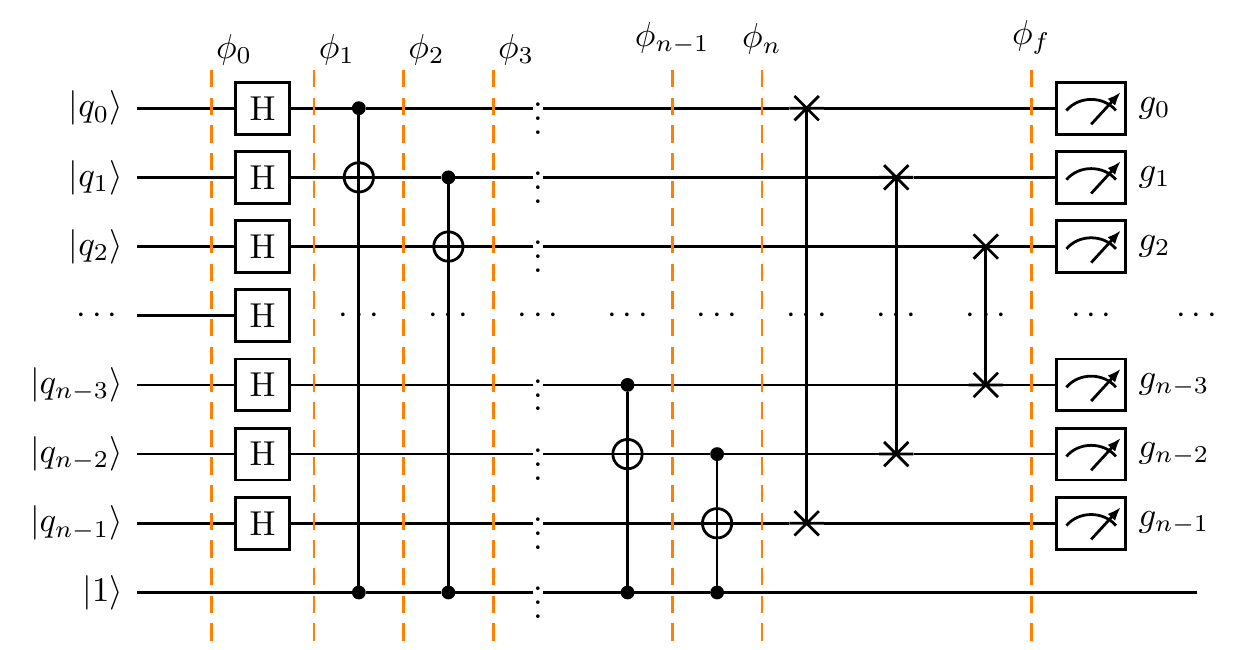}
	\caption{Quantum circuit for computing the Walsh-Hadamard transform in sequency ordering.} \label{fig:sequency}
\end{figure}

A slight variation of the quantum circuit discussed earlier in  \mref{sec:qunatum_solution} can be used for computing the Walsh-Hadamard transform in sequency order. A schematic quantum circuit for performing the Walsh-Hadamard transform in sequency order is shown in \mfig{fig:sequency}.

Since the ancilla qubit remains in the state  $ \ket{1} $ throughout the execution of the circuit, we will ignore this in the following computation and consider only the first $ n $ qubits from $ \ket{q_0} $ to $ \ket{q_{n-1}} $ (as shown in  \mfig{fig:sequency}). 
In order to show that the quantum circuit in \mfig{fig:sequency} computes the Walsh-Hadamard transforms in sequency ordering, we will compute the output when the input is a computational basis state $ \ket{\phi_0} =  \ket{j} $, with $ 0 \leq j  \leq N-1 $ and $ N =2^n $. 
We have 
\begin{align}
	\ket{\phi_1 } = H^{\otimes n} \ket{\phi_0} &=  \frac{1}{\sqrt{N}}  \sum_{s=0}^{N-1}   \, (-1)^{j \cdot s} \, \ket{s},  \nonumber \\
	&= \frac{1}{\sqrt{N}}  \sum_{s=0}^{N-1}   \,  \, (-1)^{ \sum_{k=0}^{n-1} \, j_k s_k} \, \ket{s}. \label{eq:phif}
\end{align}
Similar to our discussion in \mref{sec:algorithm}, the remaining part of the circuit acts on $ \ket{s} $ to give $ \ket{g} $, where $ g = \sum_{k=0}^{n-1} \, g_k 2^k $ and 
$ g_{n-1} = s_0 $, $  g_k = \, s_0 \, \oplus \, s_1 \, \oplus \, \ldots   \, \oplus \, s_{n-1-k}$, for $ k = n-2, \, \ldots, \, 1,\, 0 $ (see  \meqref{eq:gk}  and \meqref{eq:g}). 
Further, it is easy to see that 
\begin{equation}\label{eq:stog}
	s_k = g_{n-k} \, \oplus \, g_{n-k-1}, \quad k=0,\, 1,\, \ldots,\, n-1,  
\end{equation}
where it is assumed that $ g_n = 0 $. Also, the map sending $ \ket{s} $ to $ \ket{g} $ defined in (see \meqref{eq:gk}  and  \meqref{eq:g}) is invertible. It can be easily checked and it also follows from the fact that the transformation that sends  $ \ket{s} $ to $ \ket{g} $ is implemented by unitary quantum gates. Therefore, the summation that runs through $ s=0 $ to $ s= n-1 $ in   \meqref{eq:phif}  can be replaced by the summation running through $ g =0 $ to $ g=n-1 $. It follows that
\begin{align} \label{eq:phif_final}
	\ket{\phi_f }  &= \frac{1}{\sqrt{N}}  \sum_{g=0}^{N-1}   \, (-1)^{\sum_{k=0}^{n-1} \, j_k (g_{n-k}  \, \oplus \,  g_{n-k-1} )}   \, \ket{g}  \nonumber \\
	&= \frac{1}{\sqrt{N}}  \sum_{g=0}^{N-1}   \, (-1)^{\sum_{r=0}^{n-1} \, j_{n-1-r} (g_{r} \,  \oplus \, g_{r+1} )}   \, \ket{g}. 
\end{align}
From  \meqref{eq:sequency:hadamard} and  \meqref{eq:phif_final}  it follows that the quantum circuit shown in \mfig{fig:sequency} can be used for computing the Walsh-Hadamard transform in sequency order. 

	\section{Conclusion}\label{sec:conclusion}
		We proposed a zero-crossings counting problem that generalized the Bernstein–Vazirani problem. The problem asked to determine the number of zero-crossings (or sign changes) in the sequence $ \mathcal{S} $ defined in \meqref{eq:def_sequence}. We presented a quantum algorithm (see Algorithm~\ref{alg_Zero_Counting}) to solve the zero-crossings counting problem. The proposed quantum algorithm requires only one oracle query to determine the number of zero-crossings (or sign changes) in the sequence $ \mathcal{S} $.  In comparison, a classical algorithm would need at least $ n $ oracle queries, where $ 2^n $ is the size of the sequence $ \mathcal{S} $. 
		
		We also presented a quantum circuit for obtaining the Walsh-Hadamard transforms in sequency ordering. The Walsh-Hadamard transform in sequency ordering is used in many scientific and engineering applications, including in digital signal and image processing~\cite{kuklinski1983fast,zarowski1985spectral},  cryptography~\cite{lu2016walsh}, solution of non-linear ordinary differential equations and partial differential equations~\cite{beer1981walsh,ahner1988walsh,gnoffo2014global, gnoffo2015unsteady, gnoffo2017solutions}, optics, etc. In many image processing applications, using  the Walsh-Hadamard transforms in sequency ordering offers advantages in comparison to  the Walsh-Hadamard transforms in natural ordering. Hence, the quantum circuit shown in \mfig{fig:sequency} may find use in future quantum computing applications employing the Walsh-Hadamard transforms in sequency ordering.

%
%

				\bibliographystyle{unsrt}

	\end{document}